\documentclass[11pt]{article}

\usepackage{comment}
\usepackage{graphicx}
\usepackage{amsmath}
\usepackage{algorithm}
\usepackage{algpseudocode}
\usepackage{amssymb}
\usepackage{epstopdf}
\usepackage{graphicx}
\usepackage{epstopdf}
\usepackage{url}

\DeclareMathOperator{\range}{range}

\DeclareMathOperator{\scope}{scope}
\DeclareMathOperator{\Angle}{angle}

\newcommand{\F}{{\cal F}}

\newcommand{\T}{{\cal T}}
\newcommand{\V}{{\cal V}}

\newtheorem{definition}{Definition}
\newtheorem{property}{Property}
\newtheorem{lemma}{Lemma}
\newtheorem{theorem}{Theorem}

\newtheorem{corollary}{Corollary}

\newenvironment{proof}{\vspace{-2mm} {\bf Proof:} \rm}{\mbox{}
\hfill $\Box$ \vspace{1ex} }

\title{Optimal Monotone Drawings of Trees
}

\date{}
\author{
Dayu He\thanks{Department of Computer Science and Engineering,
State University of New York  at Buffalo, Buffalo, NY 14260, USA. Email: dayuhe@buffalo.edu.}
\ \ \
Xin He\thanks{Department of Computer Science and Engineering,
State University of New York at Buffalo, Buffalo, NY 14260, USA. Email: xinhe@buffalo.edu.
Research supported in part by NSF Grant CCR-1319732.}
}

\begin{document}
\maketitle

\begin{abstract}
A monotone drawing of a graph $G$ is a straight-line drawing of $G$
such that, for every pair of vertices $u,w$ in $G$, there exists a
path $P_{uw}$ in $G$ that is monotone in some direction $l_{uw}$. (Namely,
the order of the orthogonal projections of the vertices of $P_{uw}$
on $l_{uw}$ is the same as the order they appear in $P_{uw}$.)

The problem of finding monotone drawings for trees has been studied
in several recent papers. The main focus is to reduce the size of the
drawing. Currently, the smallest drawing size is $O(n^{1.205}) \times
O(n^{1.205})$. In this paper, we present an algorithm for
constructing monotone drawings of trees on a grid of size
at most $12n \times 12n$. The smaller drawing size is achieved by
a new simple Path Draw algorithm, and a procedure that carefully 
assigns primitive vectors to the paths of the input tree $T$.

We also show that there exists a tree $T_0$ such that any monotone drawing
of $T_0$ must use a grid of size $\Omega(n) \times \Omega(n)$.
So the size of our monotone drawing of trees is asymptotically optimal.
\end{abstract}

\section{Introduction}\label{sec:Intro}

A \textit{straight-line drawing} of a plane graph $G$
is a drawing $\Gamma$ in which each vertex of $G$ is drawn as a distinct
point on the plane and each edge of $G$ is drawn as a line segment connecting
two end vertices without any edge crossing. A path $P$ in a straight-line
drawing $\Gamma$ is \textit{monotone} if there exists a line $l$ such that
the orthogonal projections of the vertices of $P$ on $l$ appear along $l$ in
the order they appear in $P$. We call $l$ a {\em monotone line}
(or {\em monotone direction}) of $P$. $\Gamma$ is called
a \textit{monotone drawing} of $G$ if it contains at least one monotone
path $P_{uw}$ between every pair of vertices $u,w$ of $G$. We call the
monotone direction $l_{uw}$ of $P_{uw}$ the monotone direction for $u,w$.

Monotone drawing introduced by Angelini et al. \cite{ACBFP12} is a new
visualization paradigm. 
Consider the example described in \cite{ACBFP12}: a traveler
uses a road map to find a route from a town $u$ to a town $w$.
He would like to easily spot a path connecting $u$ and $w$.
This task is harder if each path from $u$ to $w$ on the map has legs
moving away from $u$. The traveler rotates the map to better perceive its
content. Hence, even if in the original map orientation all paths
from $u$ to $w$ have annoying back and forth legs, the traveler
might be happy to find one map orientation where a path
from $u$ to $w$ smoothly goes from left to right.
This approach is also motivated by human subject experiments:
it was shown that the ``geodesic tendency'' (paths following a given
direction) is important in understanding the structure of the
underlying graphs \cite{HEH2009}.

Monotone Drawing is also closely related to several other important
graph drawing problems. In a monotone drawing, each monotone path is monotone
with respect to a different line. In an {\em upward drawing} \cite{DT88,GT01},
every directed path is monotone with respect to the positive $y$ direction.
Even more related to the monotone drawings are the {\em greedy drawings}
\cite{AFG10,ML08,PR05}. In a greedy drawing, for any two vertices $u,v$,
there exists a path $P_{uv}$ from $u$ to $v$ such that the Euclidean distance 
from an intermediate vertex of $P_{uv}$ to the destination $v$ decreases at each step.
N\"{o}llenburg et al. \cite{NPR14} observed that while getting closer to the destination, a greedy
path can make numerous turns and may even look like a spiral, which hardly matches
the intuitive notion of geodesic-path tendency.
In contrast, in a monotone drawing, there exists a path $P_{uv}$ from $u$ to $v$
(for any two vertices $u,v$) and a line $l_{uv}$ such that the Euclidean
distance from the projection of an intermediate vertex of $P_{uv}$ on $l$
to the projection of the destination $v$ on $l$ decreases at each step.
So the monotone drawing better captures the notion of geodesic-path tendency.

{\bf Related works:}
Angelini et al. \cite{ACBFP12} showed that every tree of $n$ vertices
has a monotone drawing of size $O(n^2) \times O(n)$ (using a DFS-based
algorithm), or $O(n^{\log_2 3}) \times O(n^{\log_2 3})=
O(n^{1.58}) \times O(n^{1.58})$ (using a BFS-based algorithm).
It was also shown that every biconnected planar graph of $n$
vertices has a monotone drawing in real coordinate space.
Several papers have been published since then.
The focus of the research is to identify the graph classes
having monotone drawings and, if so, to find monotone drawings for them
with size as small as possible. Angelini (with another set of authors)
\cite{ADKM+13} showed that
every planar graph has a monotone drawing of size $O(n)\times O(n^2)$.
However, their drawing is not straight line.
It may need up to $4n -10$ bends in the drawing.
Recently Hossain and Rahman \cite{HR14} showed that every planar graph has a
monotone drawing.
X. He and D. He \cite{HH15b} showed that the classical Schnyder drawing
of 3-connected plane graphs on an $O(n) \times O(n)$ grid is monotone. 

The monotone drawing problem for trees is particularly important.
Any drawing result for trees can be applied to any connected
graph $G$: First, we construct a spanning tree $T$ for $G$, then
find a monotone drawing $\Gamma$ for $T$. $\Gamma$ is automatically
a monotone drawing for $G$ (although not necessarily planar).

Both the DFS- and BFS-based tree drawing algorithms in previous papers
use the so-called Stern-Brocot tree to generate a set of $n-1$ primitive
vectors (will be defined later) in increasing order of slope. Then both
algorithms do a post-order traversal of the input tree, assign
each edge $e$ a primitive vector, and draw $e$ by using the assigned vector.
Such drawings of trees are called \textit{slope-disjoint}.
Kindermann et al. \cite{KSSW14} proposed another version of the
slope-disjoint algorithm, but using a different set of
primitive vectors (based on {\em Farey sequence}), which slightly decreases
the grid size to $O(n^{1.5}) \times O(n^{1.5})$. Recently, X. He and
D. He reduced the drawing size to $O(n^{1.205}) \times O(n^{1.205})$
by using a set of more compact primitive vectors \cite{HH15a}.

A stronger version of monotone drawings is the {\em strong monotone
drawing}: For every two vertices $u,w$ in the drawing of $G$, there must
exist a path $P_{uw}$ that is monotone with respect to the line passing through
$u$ and $w$. Since the strong monotone drawing is not a subject of this paper,
we refer readers to \cite{NPR14} for related results and references. 

{\bf Our results:}
We show that every $n$-vertex tree $T$ admits a monotone drawing on a
grid of size $12n \times 12n$, which is asymptotically optimal.

The paper is organized as follows. Section \ref{sec:Pre} introduces definitions and
preliminary results on monotone drawings.  In Section \ref{sec:tree}, we give our
algorithm for constructing monotone drawings of trees on a $12n \times 12n$ grid.
In  Section \ref{sec:lower},
we describe a tree $T_0$ and show that any monotone drawing of $T_0$ must use a grid
of size $\Omega(n) \times \Omega(n)$.
Section \ref{sec:conc} concludes the paper and discusses some open problems.

\section{Preliminaries}\label{sec:Pre}

Let $p$ be a point in the plane and $l$ be a half-line with $p$ as its
starting point. The angle of $l$, denoted by $\Angle(l)$ is the
angle spanned by a ccw (we abbreviate the word ``counterclockwise''
as ccw) rotation that brings the direction of the positive $x$-axis
to overlap with $l$. We consider angles that are equivalent modulo
$360^{\circ}$ as the same angle (e.g., $270^{\circ}$ and $-90^{\circ}$
are regarded as the same angle).

In this paper, we only consider {\em straight line drawings}
(i.e., each edge of $G$ is drawn as a straight line segment between
its end vertices.) Let $\Gamma$ be such a drawing of $G$ and
let $e=(u,w)$ be an edge of $G$. The \textit{direction} of $e$,
denoted by $d(u,w)$ or $d(e)$, is the half-line starting at $u$
and passing through $w$. The angle of an edge $(u,w)$, denoted by
$\Angle(u,w)$, is the angle of $d(u,w)$. Observe that
$\Angle(u,w) = \Angle(w,u)-180^{\circ}$.
When comparing directions and their angles, we assume that they are
applied at the origin of the axes.

Let $P(u_1,u_k)=(u_1, \ldots, u_k)$ be a path of $G$. We also use $P(u_1,u_k)$
to denote the drawing of the path in $\Gamma$. $P(u_1,u_k)$ is
\textit{monotone with respect to a direction} $l$ if the orthogonal
projections of the vertices $u_1, \ldots, u_k$ on $l$ appear in the
same order as they appear along the path. $P(u_1,u_k)$ is {\em monotone}
if it is monotone with respect to some direction. A drawing $\Gamma$
is \textit{monotone} if there exists a monotone path
$P(u,w)$ for every pair of vertices $u,w$ in $G$.

\begin{figure}[!htbp]
\begin{center}
\includegraphics[width=0.95\textwidth, angle =0]{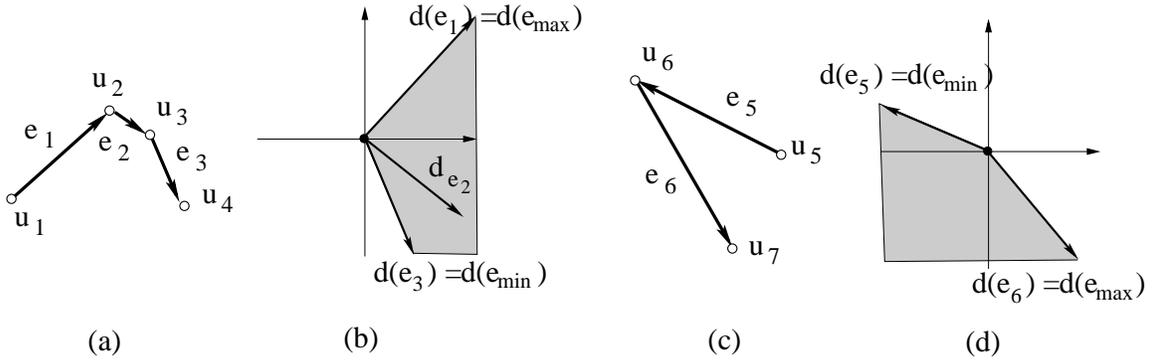}
\centering
\caption{(a) A monotone path $P(u_1,u_4)$ with extremal edges $e_1$
and $e_3$; (b) The range of $P(u_1,u_4)$ defined by $d(e_3) = d(e_{\min}) $
and $d(e_1) = d(e_{\max})$; (c) A monotone path $P(u_5,u_7)$ with only
two edges $e_5$ and $e_6$; (d) The range of $P(u_5,u_7)$ defined by
$d(e_5) = d(e_{\min}) $ and $d(e_6) = d(e_{\max})$.}
\label{Fig:path}
\end{center}
\end{figure}
\vspace{-0.1in}

The following property is well-known \cite{ACBFP12}:

\begin{property}\label{Prop:monotone}
A path $P(u_1,u_k)$ is monotone if and only if it contains two 
edges $e_1$ and $e_2$ such that the closed wedge centered at the origin
of the axes, delimited by the two half-lines $d(e_1)$ and $d(e_2)$,
and having an angle strictly smaller than $180^{\circ}$, contains
all half-lines $d(u_i, u_{i+1})$, for $i = 1, \ldots, k-1$.
\end{property}

The two edges $e_1$ and $e_2$ in Property \ref{Prop:monotone}
are called the two {\em extremal edges} of $P(u_1,u_k)$, and
the closed wedge (centered at the origin of the axes) delimited by
the two half-lines $d(e_1)$ and $d(e_2)$, containing all the half-lines
$d(u_i,u_{i+1})$ for $i=1, \ldots, k-1$, is called the \textit{range}
of $P(u_1,u_k)$ and denoted by $\range(P(u_1,u_k))$.
See Fig \ref{Fig:path} (a) and (b).
We use $e_{\min}$ and $e_{\max}$ to denote the extremal
edges $e_1$ and $e_2$ so that the wedge $\range(P(u_1,u_k))$
is the area spanned by a ccw rotation that brings
the half-line $d(e_{\min})$ to overlap with the half-line $d(e_{\max})$.
Thus we have $\Angle(e_{\min}) < \Angle(e_{\max})$.
Note that, for a path with only two edges, we consider its range to
be the closed wedge with an angle $\leq 180^{\circ}$.
See Fig \ref{Fig:path} (c) and (d). 

The closed interval $[\Angle(e_{\min}),\Angle(e_{\max})]$ is called
the \textit{scope} of $P(u_1,u_k)$ and denoted by $\scope(P(u_1,u_k))$.
Note that $\Angle(u_i,u_{i+1}) \in \scope(P(u_1,u_k))$ for all edges
$(u_i,u_{i+1})$ ($i=1, \ldots, k-1$) in $P(u_1,u_k)$.
By this definition, Property \ref{Prop:monotone} can be restated as:

\begin{property}\label{Prop:monotone-1}
A path $P(u_1,u_k)$ with scope $[\Angle(e_{\min}),\Angle(e_{\max})]$
is monotone if and only if 
$\Angle(e_{\max})- \Angle(e_{\min}) < 180^{\circ}$. 
\end{property}

Define: $P_d =\{(x,y)~|~x \mbox{ and } y \mbox{ are integers, }\mbox{gcd}(x,y)=1,
1 \leq x \leq y \leq d\}$

If we consider each entry $(x,y) \in P_d$ to be the rational number $y/x$ and
order them by value, we get the so-called {\em Farey sequence} $\F_d$
(see \cite{HW89}). The property of the Farey sequence is well understood.
It is known $|\F_d| = 3d^2/\pi^2 + O(d\log d)$ (\cite{HW89}, Thm 331).
Thus, $|P_d|=|\F_d| \geq  3d^2/\pi^2$.
Let $P'_d$ be the set of the vectors that are the reflections of the vectors
in $P_d$ through the line $x=y$. Define:
\vspace{-5pt}
\[\overline{P_d} = P_d \cup P'_d=\{(x,y)~|~x \mbox{ and } y \mbox{ are integers, }
\mbox{gcd}(x,y) =1, 1 \leq x, y \leq d\}
\]

The members of $\overline{P_d}$ are called the {\em primitive vectors
of size} $d$. Fig \ref{Fig:example} (a) shows the vectors in
$\overline{P_3}$. We have $|\overline{P_d}| \geq 6d^2/\pi^2$.
Moreover, the members of $\overline{P_d}$ can be enumerated in
$O(|\overline{P_d}|)$ time \cite{KSSW14}.
Note that the vectors $(1,0)$ and $(0,1)$ are not vectors in $\overline{P_d}$.
For easy reference, we call them the {\em boundary vectors} of $\overline{P_d}$.

\begin{figure}[!htbp]
\begin{center}
\includegraphics[width=0.85\textwidth, angle =0]{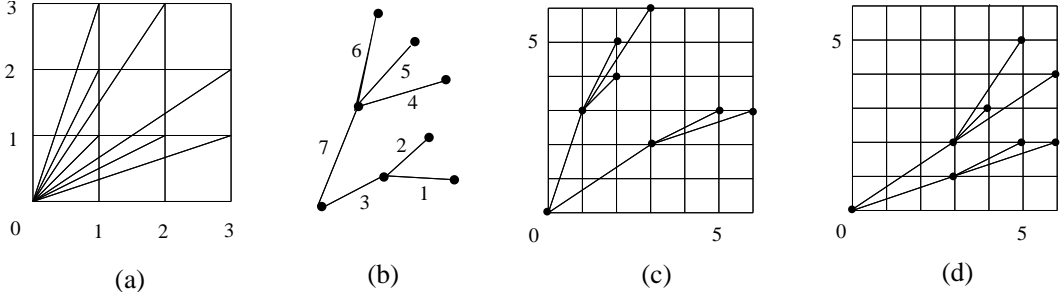}
\centering
\caption{(a) The vectors in $\overline{P_3}$; (b) a tree with edges ordered in ccw
post-order; (c) the monotone drawing of the tree in (b) produced by
Algorithm \ref{alg:draw-tree} by using vectors in (a);
(d) the monotone drawing of the tree in (b) produced by Algorithm
\ref{alg:path} by using vectors in (a).}
\label{Fig:example}
\end{center}
\end{figure}

Next, we outline the algorithm in \cite{ACBFP12} for monotone drawings of trees.

\begin{definition} \cite{ACBFP12} \label{def:disjoint}
A {\em slope-disjoint drawing} of a rooted tree $T$ is such that:
\begin{enumerate}
\item For each vertex $u$ in $T$, there exist two angles $\alpha_1(u)$
and $\alpha_2(u)$, with $0< \alpha_1(u) < \alpha_2(u) < 180^{\circ}$
such that, for every edge $e$ that is either in $T(u)$ ($T(u)$ denotes
the subtree of $T$ rooted at $u$) or that connects
$u$ with its parent, it holds that $\alpha_1(u) < \Angle(e) < \alpha_2(u)$;

\item for any vertex $u$ in $T$ and a child $v$ of $u$, it holds
that $\alpha_1(u) < \alpha_1(v) < \alpha_2(v) < \alpha_2(u)$;

\item for every two vertices $v_1,v_2$ with the same parent, it holds
that either $\alpha_1(v_1) < \alpha_2(v_1) < \alpha_1(v_2) < \alpha_2(v_2)$ or
$\alpha_1(v_2) < \alpha_2(v_2) < \alpha_1(v_1) < \alpha_2(v_1)$.
\end{enumerate}
\end{definition}

The following theorem was proved in \cite{ACBFP12}.

\begin{theorem}\label{thm:disjoint}
Every slope-disjoint drawing of a tree is monotone.
\end{theorem}

Remark: By Theorem \ref{thm:disjoint}, as long as the angles of the edges
in a drawing of a tree $T$ guarantee the slope-disjoint property, one can
{\em arbitrarily} assign lengths to the edges always obtaining a monotone
drawing of $T$.

Algorithm \ref{alg:draw-tree} for producing monotone drawing of
trees was described in \cite{ACBFP12}. (The presentation here
is slightly modified).

\begin{algorithm}
\caption{Tree-Monotone-Draw}\label{alg:draw-tree}
\begin{description}
\item[Input:] A tree $T=(V,E)$ with $n$ vertices.
\item[1.] Take any set $\V=\{(x_1,y_1), \ldots, (x_{n-1},y_{n-1})\}$
of $n-1$ distinct primitive vectors, sorted by increasing $y_i/x_i$
value.
\item[2.] Assign the vectors in $\V$ to the edges of $T$
in ccw post-order.
\item[3.] Draw the root $r$ of $T$ at the origin point $(0,0)$.
Then draw other vertices of $T$ in ccw pre-order as follows:
\item[3.1] Let $w$ be the vertex to be drawn next;
let $u$ be the parent of $w$ which has been drawn at the point $(x(u),y(u))$.
\item[3.2] Let $(x_i,y_i)$ be the primitive vector assigned
to the edge $(u,w)$ in step 2. Draw $w$ at the point
$(x(w),y(w))$ where $x(w)=x(u)+x_i$ and $y(w)=y(u)+y_i$.
\end{description}
\end{algorithm}

Fig \ref{Fig:example} (b) shows a tree $T$. The numbers next to the edges
indicate the order they are assigned the vectors in $\V=\overline{P_3}$. 
Fig \ref{Fig:example} (c) shows the drawing of $T$ produced by
Algorithm \ref{alg:draw-tree}.

It was shown in \cite{ACBFP12} that the drawing obtained in
Algorithm \ref{alg:draw-tree} is slope-disjoint and hence monotone.
Two versions of Algorithm \ref{alg:draw-tree} were given in \cite{ACBFP12}.
Both use the Stern-Brocot tree $\T$ to generate the vector set $\V$
needed in step 1. The BFS version of the algorithm
collects the vectors from $\T$ in a breath-first-search fashion.
This leads to a drawing of size $O(n^{\log_2 3}) \times O(n^{\log_2 3})
= O(n^{1.58}) \times O(n^{1.58})$.
The DFS version of the algorithm collects the vectors from $\T$ in a
depth-first-search fashion. This leads to a drawing of size $O(n) \times O(n^2)$.
The algorithm in \cite{KSSW14} for finding monotone drawings of trees is
essentially another version of Algorithm \ref{alg:draw-tree}. It uses
the vectors in $\overline{P_d}$ (with $d = 4 \sqrt{n}$) for the set $\V$ in step 1.
This leads to a monotone drawing of size $O(n^{1.5}) \times O(n^{1.5})$.
The algorithm in \cite{HH15a} uses a more careful vector assignment procedure.
This reduces the drawing size to $O(n^{1.205}) \times O(n^{1.205})$.

\section{Monotone Drawings of Trees on a $12n \times 12n$ Grid}
\label{sec:tree}

In this section, we describe our algorithm for optimal
monotone drawings of trees.

\subsection{Path Draw Algorithm}\label{subsection:draw-tree}

In this subsection, we present a new {\em Path Draw Algorithm}
for constructing monotone drawings of trees.
It follows the same basic ideas of Algorithm \ref{alg:draw-tree},
but will allow us to produce a monotone drawing with size $O(n) \times O(n)$.

Let $T$ be a tree rooted at $r$ with $t$ leaves. To simplify notations,
let $L=\{ 1, 2, \ldots, t\}$ denote the set of the leaves in $T$ in
ccw order. (For visualization purpose, we draw the root of $T$ at
the top and also refer ccw order as left to right order in this paper). 

\begin{definition}\label{def:path}
Let $L_{\sigma} = \{l_1,...,l_t\}$ be any permutation of the leaf set $L$ of $T$.
The {\em path decomposition with respect to $L_{\sigma}$} is the partition
of the edge set of $T$ into $t$ edge-disjoint paths, denoted by
$B_{\sigma} = \{b_1, b_2, \ldots, b_t\}$ defined as follows.
\begin{itemize}
\item $b_1$ is the path from $l_1$ to the root $r$ of $T$.
\item Suppose $b_1, \ldots, b_k$ have been defined. Let $T_k = \cup_{i=1}^k b_i$.
Let $p_{k+1}$ be the path from $l_{k+1}$ to $r$ and let $u$ be the first
vertex in $p_{k+1}$ that is also in $T_k$. Define $b_{k+1}$ as the sub-path
of $p_{k+1}$ between $l_{k+1}$ and $u$. ($u$ is called the {\em attachment of}
$l_{k+1}$ in $T_k$). 
\end{itemize}
\end{definition}

\begin{figure}[!htb]
\begin{center}
\includegraphics[width=0.55\textwidth, angle =0]{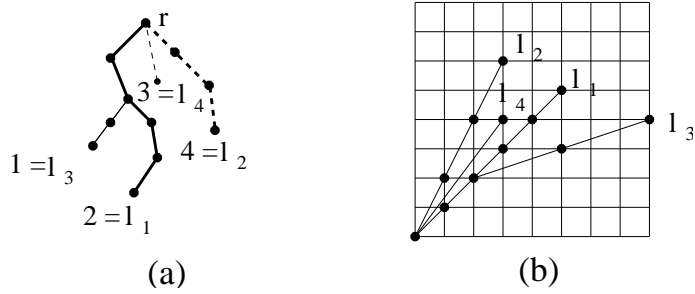}
\centering
\caption{(a) The path decomposition of a tree $T$ with respect to the
leaf-permutation $(2,4,1,3)$; (b) the drawing of $T$
produced by Algorithm \ref{alg:path} by using the vector set
$\V=\{1/3, 1/1, 4/3, 2/1 \}$.}
\label{Fig:path-draw}
\end{center}
\vspace{-0.1in}
\end{figure}

Fig \ref{Fig:path-draw} (a) shows a path decomposition of a tree with 4 leaves.
The thick solid, thick dotted, thin solid and thin dotted lines correspond to $b_1, b_2, b_3$ and $b_4$, respectively. By a slight abuse of
notation, we also use $b_k$ to denote the set of the vertices in path $b_k$ except
the last vertex (the attachment vertex of $b_k$ in $T_k$).
With this understanding, $\{ b_k~|~1 \leq k \leq t\}$ is also a partition
of the vertex set of $T$.

Our Path Draw Algorithm is described in Algorithm \ref{alg:path}.

\begin{algorithm}
\caption{Path Draw Algorithm}\label{alg:path}
\begin{description}
\item[Input:] A tree $T=(V,E)$ and a set of paths $B_{\sigma}=\{b_1,\ldots,b_t\}$
of $T$ with respect to a permutation $L_{\sigma}$ of the leaves in $T$.

\item[1.] Take any set $\V=\{(x_1,y_1), \ldots, (x_t,y_t)\}$
of $t$ distinct primitive vectors, sorted by increasing $y_i/x_i$
value.
\item[2.] Assign the vectors in $\V$ to the leaves of $T$ in ccw order. 
(For example, in Fig \ref{Fig:path-draw} (a), the leaves $l_3, l_1, l_4, l_2$
are assigned the 1st, 2nd, 3rd and the 4th vector in $\V$.)

\item[2.1.] Let $(x,y)$ be the vector assigned to $l_j$ in step 2.
Assign the vector $(x,y)$ to all edges in $b_j$. (We say the vector
$(x,y)$ is assigned to the path $b_j$). Do the same for every path
$b_j$ in $B$. Every edge in $T$ is assigned a vector in $\V$ by now.

\item[3.] Draw the vertices of $T$ as in step 3 of Algorithm \ref{alg:draw-tree}.
\end{description}
\end{algorithm}

Fig \ref{Fig:example} (d) shows the drawing of the tree in Fig
\ref{Fig:example} (a) produced by Algorithm \ref{alg:path}
(by using the vectors in $\overline{P_3}$, and the permutation that lists the
leaves in ccw order). Fig \ref{Fig:path-draw} (b)
shows the drawing of the tree in Fig \ref{Fig:path-draw}
(a) by Algorithm \ref{alg:path} by using the vector set
$\V=\{1/3, 1/1, 4/3, 2/1 \}$.

\begin{theorem}\label{thm:path}
Algorithm \ref{alg:path} produces a monotone drawing of a tree $T$
for any permutation $L_{\sigma}$ of the leaves of $T$.
\end{theorem}

\begin{proof}
Consider two vertices $i,j$ in $T$. Let $P_{ij}$ be the (unique) path in
the drawing of $T$ from $i$ to $j$. We need to show $P_{ij}$ is
a monotone path. If either $i$ is an ancestor of $j$, or $j$ is an ancestor
of $i$, this is trivially true (because every edge in $P_{ij}$ has angle
between $0^{\circ}$ and $90^{\circ}$). So we assume this is not the case.

Let $u$ be the lowest common ancestor of $i$ and $j$ in $T$.
Since any subpath of a monotone path is monotone \cite{ACBFP12}, without
loss of generality, we assume both $i$ and $j$ are leaves of $T$, and $i$
is located to the left of $j$ ($i$ appears before $j$ in ccw order)
in the drawing. See Fig \ref{Fig:decomposition} (a).

Let $P_{iu}$ ($P_{uj}$, respectively) be the subpath of $P_{ij}$
from $i$ to $u$ (from $u$ to $j$, respectively). Consider any edge $e_a \in P_{iu}$ and
any edge $e_b \in P_{uj}$. Let $\overline{e_a}$ be the edge $e_a$ but in
opposite direction (i.e., directed away from the root).
Then, $\overline{e_a}$ belongs to a path $b_{i'}$ and
$e_b$ belongs to a path $b_{j'}$ in the path decomposition.
It is easy to see that $b_{i'}$ must appear to the left of $b_{j'}$.
Thus:
$\Angle(\overline{e_a}) < \Angle(e_b) \mbox{ and } \Angle(e_a) 
	= \Angle(\overline{e_a}) + 180^{\circ}$.

	Let $e_{\min}$ and $e_{\max}$ be the two extremal edges in $P_{ij}$ with
	$\Angle(e_{\min}) < \Angle(e_{\max})$. Then, $e_{\max}$ must be an edge $e_a$ in $P_{iu}$
	and $e_{\min}$ must be an edge $e_b$ in $P_{uj}$. By the above equation and inequality,
	we have:
	$\Angle(e_{\max}) - \Angle(e_{\min}) = \Angle(\overline{e_a}) + 180^{\circ}
	- \Angle(e_b) < 180^{\circ}$.
	By Property \ref{Prop:monotone-1}, $P_{ij}$ is a monotone path as to be shown.
\end{proof}

\begin{figure}[!htb]
\begin{center}
\includegraphics[width=0.55\textwidth, angle =0]{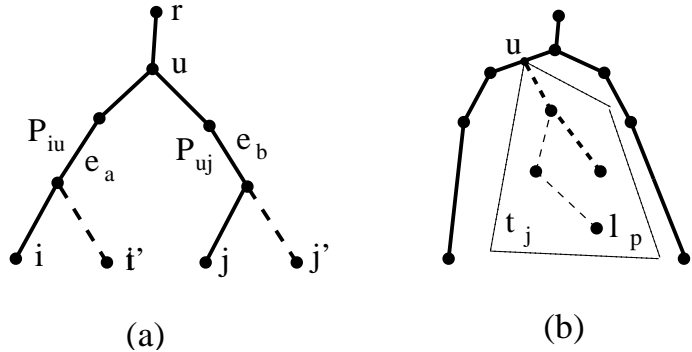}
\centering
\caption{(a) The proof of Theorem \ref{thm:path}; (b) the proof of Lemma \ref{lemma:decompose}.}
\label{Fig:decomposition}
\end{center}
\vspace{-0.2in}
\end{figure}

\subsection{Length Decreasing Path Decomposition of $T$}

In this subsection, we define a special path decomposition of $T$, called the
{\em length decreasing path decomposition} and denoted by LDPD. Later, we will
apply Algorithm \ref{alg:path} with respect to this decomposition.
Note that LDPD is a special case of the well-known
\emph{heavy path decomposition} \cite{ST83}. However,
	our algorithm does not need any operation provided
	by \emph{heavy path decomposition}, only the definition. 

	\begin{definition}\label{def:LDPD}
	A {\em length decreasing path decomposition} $B = \{b_1,b_2,...,b_t\}$ 
	of a tree $T$ is defined as follows:

	\begin{enumerate}
	\item Let $l_1$ be the vertex that is the farthest from the root $r$ of $T$
	(break ties arbitrarily). Define $b_1$ as the path from $l_1$ to $r$.

	\item Suppose that the leaves $l_1, \ldots, l_k$ and the corresponding
	paths $b_1, \ldots, b_k$ have been defined.
	Let $T_k=\cup_{i=1}^k b_i$.
	Let $l_{k+1}$ be the leaf in $L-\{l_1, \ldots,l_k\}$ such that
	the path from $l_{k+1}$ to its attachment $u$ in $T_k$ is the longest among all
	leaves in $L-\{l_1, \ldots,l_k\}$ (break ties arbitrarily).
	Define $b_{k+1}$ as the path from $l_{k+1}$ to $u$.
	\end{enumerate}
	\end{definition}

	Fig \ref{Fig:path-draw} (a) shows a LDPD of a tree.
	Let $|b_i|$ denote the number of edges in $b_i$. By definition, we have
	$|b_i|\geq |b_{i+1}|$ for $1 \leq i < t$. We further
	partition the paths in $B$ as follows.

	\begin{definition}\label{Def:decomposition}
	Let $T$ be a $n$-vertex tree and $B = \{b_1, b_2,..., b_t\}$ be a LDPD of $T$. 
	Let $c>1$ be an integer and $K=\lceil \log_c n\rceil$. The {\em $c$-partition
		of $B$} is a partition of $B$ defined as:
		\begin{eqnarray*}
		D_1 & = & \{ b_i \in B~|~ |b_i| \in [\frac{n-1}{c},(n-1)] \}\\
			D_j & = & \{ b_i \in B~|~ |b_i| \in [\frac{n-1}{c^{j}},
		\frac{n-1}{c^{j-1}})\}, \mbox{ for }1 < j \leq K
		\end{eqnarray*}
		\end{definition}

		Note that $D_j$'s are disjoint and $\cup_{j=1}^K D_j = B$. 
		Let $m_j=|D_j|$ ($1\leq j \leq K$) be the number of paths in $D_j$.
		We have the following:

		\begin{property}\label{prop:partition}
		$c^{K-1} \leq \sum_{j=1}^K m_j c^{K-j} \leq c^K$.
		\end{property}

		\begin{proof}
		For each $j$ ($1 \leq j\leq K$), there are $m_j$ paths in $D_j$. Each path $b_l \in D_j$
		contains $|b_l| \in [\frac{n-1}{c^j},\frac{n-1}{c^{j-1}})$ edges. Thus we have:
		\[ \sum_{j=1}^K m_j \frac{n-1}{c^j} \leq n-1 = \sum_{j=1}^K \sum_{{b_l} \in D_j} |b_l|
		\leq \sum_{j=1}^K m_j \frac{n-1}{c^{j-1}}
		\]
		Hence $\sum_{j=1}^K \frac{m_j}{c^j} \leq 1 \leq
		\sum_{j=1}^K \frac{m_j}{c^{j-1}}$. This implies the property.
		\end{proof}

		For each $j$,  let $T[D_j]$ denote the subgraph of $T$ induced by the edge
		set $\cup_{b_l \in D_j} b_l$. $T[D_j]$ may consist of several
		subtrees in $T$. We call the subtrees in $T[D_j]$ the {\em $j$-level}
		subtrees of $T$. We have: 

		\begin{lemma}\label{lemma:decompose}
		For any $j$ ($1 \leq j \leq K$), let $t_j$ be the subtree among all $j$-level
		subtrees with the largest height $h_j$. Then, $h_j < \frac{n-1}{c^{j-1}}$.
		\end{lemma}

\begin{proof}
For a contradiction, suppose $h_j \geq \frac{n-1}{c^{j-1}}$.
Let $l_p$ be the leaf in $t_j$ with the largest distance
to the attachment $u$ of $t_j$. So the length of the path from $l_p$
to $u$ is $h_j \geq \frac{n-1}{c^{j-1}}$. See Fig \ref{Fig:decomposition} (b).
By the definition of the LDPD, $l_p$ should have been chosen as $l_q$
for some index $q < p$ such  that $|b_q| \geq \frac{n-1}{c^{j-1}}$. This
contradiction shows the assumption $h_j \geq \frac{n-1}{c^{j-1}}$ is false. 
\end{proof}

\subsection{Vector Assignment}\label{subsec:vector}

We will use Algorithm \ref{alg:path} with respect to a LDPD 
$B=\{ b_1,\ldots, b_t\}$ of $T$. First, we need the following definition:

\begin{definition}\label{def:valid}
Two positive integers $(f,d)$ are called a {\em valid-pair}
if the following hold:
\begin{enumerate}
\item $f \geq d$;
\item For any positive integer $\Delta$ and any two consecutive vectors
$(x_1,y_1)$ and $(x_2,y_2)$ in $\overline{P_{\Delta}}$ with
$y_1/x_1 < y_2/x_2$ (either one can be the boundary vector
(0,1) or (1,0)), there exist at least $f$ vectors $(x,y)$
in $\overline{P_{d\Delta}}-\overline{P_{\Delta}}$
such that $y_1/x_1 < y/x < y_2/x_2$.
\end{enumerate}
\end{definition} 

Our algorithm works for any valid-pair $(f,d)$. Later we will show
$(f,d)=(3,3)$ is a valid-pair. The main ideas of our algorithm are as follows.
Take a valid-pair $(f,d)$.
Set $c=f+1$ and let $K=\lceil \log_c n\rceil$.
Let $T$ be a tree and $B=\{b_1,\ldots, b_t\}$ be a LDPD of $T$.
Let $D = \{D_1, D_2,..., D_K\}$ be the $c$-partition of $B$. The vectors in $\overline{P_{d^{j}}} - \overline{P_{d^{j-1}}}$
are called {\em level}-$j$ vectors (for $1\leq j \leq K$). In other words,
the level-$j$ vectors are the vectors with size in the range $(d^{j-1},d^j]$.
We assign the level-1 vectors to the paths in $D_1$.
(Because the paths in $D_1$ are very long, we assign very short level-1 vectors
to them). As the index $j$ becomes larger, the paths in $D_j$ are shorter.
We can afford to assign longer level-$j$ vectors to the paths in $D_j$
without increasing the size of the drawing too much.

Next, we describe our algorithm in details. It first constructs a set
$\V$ of primitive vectors as follows.

\begin{itemize}
\item There is only one primitive vector $(1,1)$ in $\overline{P_1}$.
By the definition of the valid-pair, there exist at least
$f$ vectors in $\overline{P_{d}}- \overline{P_{1}}$ between
$(1,0)$ and $(1,1)$. Let $S_1$ be a set of $f$ vectors among them.
Similarly, there exist at least
$f$ vectors in $\overline{P_{d}}- \overline{P_{1}}$ between
$(1,1)$ and $(0,1)$. Let $S_2$ be a set of $f$ vectors among them.
Define $R_1=S_1 \cup S_2 \cup \{ (1,1)\}$. Thus $|R_1|= 2f+1$.

\item Between any two consecutive vectors in $R_1\cup \{(0,1), (1,0)\}$,
there are at least $f$ vectors in $\overline{P_{{d}^2}} - \overline{P_{{d}^1}}$.
Pick exactly $f$ vectors among them. Let $R_2$ be the union of all vectors
picked for all consecutive pairs. Thus $|R_2| = (2f+2)f$.

\item Suppose we have defined $R_1,\ldots, R_j$. Between any two consecutive
vectors in $[\cup_{i=1}^j R_i] \cup \{ (0,1), (1,0)\}$, there are
at lease $f$ vectors in $\overline{P_{{d}^{j+1}}} - \overline{P_{{d}^j}}$.
Pick exactly $f$ vectors among them. Let $R_{j+1}$ be the union of all vectors
picked for all consecutive pairs. Thus $|R_{j+1}| = (1+\sum_{i=1}^j|R_i|)f$.

\item Define $\V = \cup_{j=1}^K R_j$ (in increasing order of slopes).
\end{itemize}

\begin{lemma}\label{lemma:size}
\begin{enumerate}
\item For any $p$ ($1\leq p \leq K$), $\sum_{j=1}^p |R_j| = 2(f+1)^p -1$.
(This implies $|\V| = 2(f+1)^K -1$).

\item Let $j$ ($1 \leq j \leq K$) be an integer. Consider any vector $V_{\alpha}$
in $\V$. Let $V_{\beta}$ be the first vector in $\V$ after $V_{\alpha}$ that is in
$R_i$ with $i \leq j$. Then there are at most $(f+1)^{K-j}-1$ vectors
in $\V$ between $V_{\alpha}$ and $V_{\beta}$.
\end{enumerate}
\end{lemma}
\begin{proof}
Statement 1: We prove the equality by induction on $p$.

When $p=1$, $|R_1| = 2f+1 = 2(f+1) -1$ is trivially true.

Assume $\sum_{j=1}^p |R_j| = 2(f+1)^p -1$.

Then: $\sum_{j=1}^{p+1} |R_j| = (\sum_{j=1}^p |R_j|) + |R_{p+1}| 
= (2(f+1)^p-1) + (2(f+1)^p -1 +1)f =2(f+1)^{p+1} -1$ as to be shown.

Statement 2: Let $S$ be the set of the vectors in $\V$ that are between
$V_{\alpha}$ and $V_{\beta}$. The worst case (that $S$ has the largest size)
is when $V_{\alpha}$ itself is a level-$j$ vector.
In this case, $S$ contains $f$ level-$(j+1)$ vectors,
$(f+1)f$ level-$(j+2)$ vectors, $\ldots$ and so on. By using induction similar
to the proof of Statement 1, we can show $|S| = (f+1)^{K-j}-1$.
\end{proof}

Let $b_1, \ldots, b_t$ be the paths in a LDPD of $T$ ordered from left to right.
We call $b_l$ a {\em level}-$j$ path if $b_l \in D_j$. We will assign
the vectors in $\V$ to the paths $b_l$ ($1\leq l \leq t$) such that
the following properties hold:
\begin{itemize}
\item The vectors in $\V$ (in the order of increasing slopes)
are assigned to $b_l$'s in ccw order.

\item For each level-$j$ path $b_l$, $b_l$ is assigned a vector in $R_i$
with $i \leq j$.
\end{itemize}

\subsection{Algorithm}

Now we present our optimal monotone drawing Algorithm for trees.

\begin{algorithm}[!htb]
\caption{Optimal Draw}\label{alg:final}
\begin{description}
\item[Input:] A tree $T=(V,E)$, and a valid-pair ($f, d$).

\item[1.] Find a LDPD $B=\{ b_1,\ldots b_t\}$ of $T$ (ordered from left to right).

\item[2.] Set $c=f+1$ and let $K=\lceil \log_c n\rceil$.
Construct the $c$-partition $D=\{D_1, \ldots, D_K\}$ of $B$.

\item[3.] Let $\V$ be the set of primary vectors (in increasing
order of slops) defined in subsection \ref{subsec:vector}.

\item[4.] For $l=1$ to $t$ do:

If $b_l$ is a level-$j$ path, assign the next available (skip some vectors
in $\V$ if necessary) vector in $\V$ that is in $R_i$ with $i \leq j$.

\item[5.] Draw the vertices of $T$ as in step 3 of Algorithm \ref{alg:draw-tree}.
\end{description}
\end{algorithm}

It is not clear whether there are enough vectors in $\V$ such that
the vector assignment procedure described in Algorithm \ref{alg:final}
can succeed. The following lemma shows this is indeed the case.

\begin{lemma}\label{lemma:assign}
There are enough vectors in $\V$ such that the vector assignment procedure 
in Algorithm \ref{alg:final} can be done.
\end{lemma}

\begin{proof}
Consider a level-$j$ path $b_l$. In order to assign a vector
$V_{\beta} \in \V$ to $b_l$, we may have to skip at most $(f+1)^{K-j} -1$
vectors in $\V$ by Lemma \ref{lemma:size}.
Also counting the vector $V_{\beta}$ assigned to $b_l$, we {\em consume} at most
$(f+1)^{K-j}$ vectors in $\V$. Thus the total number
of vectors needed by the vector assignment procedure is bounded by:

\[ \sum_{j=1}^K m_j\cdot (f+1)^{K-j} =\sum_{j=1}^K m_j\cdot c^{K-j}  
 \leq c^K \mbox{ (by Property \ref{prop:partition})} 
\leq 2c^K -1 = |\V| \mbox{ (by Lemma \ref{lemma:size}.)}
\]
Thus, there are enough vectors in $\V$ for the vector assignment procedure in Algorithm \ref{alg:final}.
\end{proof}

Now we can prove our main theorem:

\begin{theorem}\label{thm:main}
For any valid-pair $(f,d)$, Algorithm \ref{alg:final} constructs a
monotone drawing of $T$ with size $I \times I$, where
$I\leq \frac{(f+1)\cdot d}{(f+1)-d} n$, in $O(n)$ time.
\end{theorem}

\begin{proof}
Because the vector assignments in Algorithm \ref{alg:final} satisfy
the condition required by Algorithm \ref{alg:path}, it indeed produces
a monotone drawing of $T$ by Theorem \ref{thm:path}. It's straightforward
to show that the algorithm takes $O(n)$ time by using basic algorithmic
techniques as in \cite{ACBFP12,HH15a}. Next, we analyze the size of
the drawing.

The subgraph $T[D_1]$ is a subtree of $T$ with height at most $n-1<n$.
The paths in $D_1$ are assigned the vectors of length at most $d$. 
So, Algorithm \ref{alg:final} draws the level-1 subtree $T[D_1]$ on a
grid of size at most $d\cdot n \times d \cdot n$.

In general, the height of any subtree in the subgraph $T[D_j]$ is at most
$(n-1)/c^{j-1} < n/c^{j-1}$ (where $c=f+1)$ by Lemma \ref{lemma:decompose}.
The paths in $D_j$ are assigned the vectors of length at most $d^j$. 
So Algorithm \ref{alg:final} can draw the level-$j$ subtrees in $T[D_j]$,
increasing the size of the drawing by at most $d \cdot (d/c)^{j-1} \cdot n$ in
both $x$- and $y$-direction.

So, Algorithm \ref{alg:final} draws $T$ on an $I \times I$ grid, where
$I \leq n \cdot d \cdot \sum_{j=1}^K (d/c)^{j-1} < \frac{c\cdot d}{c-d} n
=\frac{(f+1)\cdot d}{(f+1)-d} n$.
\end{proof}

\subsection{The Existence of Valid-pairs}

Let $F_0=1, F_1=1, F_2=F_0+F_1=2, \ldots, F_{i+2}=F_{i+1}+F_i \ldots$ be
the Fibonacci numbers. In this subsection, we show:
\begin{lemma}\label{lemma:valid}
For any integer $q\geq 2$, $(2^q-1,F_{q+1})$ is a valid-pair.
\end{lemma}

\begin{proof}
Fix a positive integer $\Delta$. Let $y_1/x_1$ and $y_2/x_2$ be any two
consecutive vectors in $\overline{P_{\Delta}}$. We have $y_2 x_1 - y_1 x_2 = 1$
(\cite{HW89}, Theorem 28) and $x_1+x_2 > \Delta$ (\cite{HW89}, Theorem 30).

Define an operator $\odot$ of two fractions as follows:

\[\frac{y_1}{x_1} \odot \frac{y_2}{x_2} = \frac{y_1+y_2}{x_1+x_2}
\]

Let $y_3/x_3=\frac{y_1}{x_1} \odot \frac{y_2}{x_2}$.
It is easy to show that $y_3/x_3$ is a fraction strictly between $y_1/x_1$ and
$y_2/x_2$. Similarly, let $y_4/x_4 = y_1/x_1 \odot y_3/x_3$
and $y_5/x_5 =y_3/x_3 \odot y_2/x_2$, we have three fractions
$y_4/x_4 < y_3/x_3 <y_5/x_5$ strictly between $y_1/x_1$ and $y_2/x_2$.

Repeating this process, we can generate all fractions between $y_1/x_1$ and
$y_2/x_2$ in the form of a binary tree, called the Stern-Brocot tree for
$y_1/x_1$ and $y_2/x_2$, denoted by $\T(y_1/x_1, y_2/x_2)$, as follows.
(The original Stern-Brocot tree defined in \cite{Br1860,St1858} is for
the fractions $y_1/x_1=0/1$ and $y_2/x_2=1/0$).

$\T(y_1/x_1, y_2/x_2)$ has two nodes $y_1/x_1$ and $y_2/x_2$ in level 0.
Level 1 contains a single node $r$ labeled by the fraction
$y_3/x_3=y_1/x_1 \odot y_2/x_2$, which is the right child of $y_1/x_1$,
and is the left child of $y_2/x_2$. An infinite ordered binary tree rooted at
$y_3/x_3$ is constructed as follows. Consider a node $y/x$ of the tree.
The left child of $y/x$ is $y/x \odot y'/x'$ where $y'/x'$ is the ancestor
of $y/x$ that is closest to $y/x$ (in terms of graph-theoretical distance in
$\T(y_1/x_1, y_2/x_2)$) and that has $y/x$ in its right subtree.
The right child of $y/x$ is $y/x \odot y''/x''$ where $y''/x''$ is the ancestor
of $y/x$ that is closest to $y/x$ and that has $y/x$ in its left subtree.
(Fig \ref{fig:SB-tree} shows a portion of the Stern-Brocot
tree $\T(4/5,5/6)$. The leftmost column indicates the level numbers).

\begin{figure}[!htb]
\begin{center}
\begin{picture}(400,130)(25,80)
\put(160,200){\makebox(0,0){$\frac{4}{5}$}}
\put(240,200){\makebox(0,0){$\frac{5}{6}$}}
\put(165,193){\line(3,-1){30}}
\put(235,193){\line(-3,-1){30}}
\put(200,175){\makebox(0,0){$\frac{9}{11}$}}
\put(120,150){\makebox(0,0){$\frac{13}{16}$}}
\put(280,150){\makebox(0,0){$\frac{14}{17}$}}
\put(190,170){\line(-4,-1){60}}
\put(210,170){\line(4,-1){60}}
\put(80,125){\makebox(0,0){$\frac{17}{21}$}}
\put(160,125){\makebox(0,0){$\frac{22}{27}$}}
\put(240,125){\makebox(0,0){$\frac{23}{28}$}}
\put(320,125){\makebox(0,0){$\frac{19}{23}$}}
\put(112,145){\line(-5,-2){27}}
\put(128,145){\line(5,-2){27}}
\put(272,145){\line(-5,-2){27}}
\put(288,145){\line(5,-2){27}}
\put(55,95){\makebox(0,0){$\frac{21}{26}$}}
\put(105,95){\makebox(0,0){$\frac{30}{37}$}}
\put(135,95){\makebox(0,0){$\frac{35}{43}$}}
\put(185,95){\makebox(0,0){$\frac{31}{38}$}}
\put(215,95){\makebox(0,0){$\frac{32}{39}$}}
\put(265,95){\makebox(0,0){$\frac{37}{45}$}}
\put(295,95){\makebox(0,0){$\frac{33}{40}$}}
\put(345,95){\makebox(0,0){$\frac{24}{29}$}}
\put(72,120){\line(-1,-2){10}}
\put(88,120){\line(1,-2){10}}
\put(152,120){\line(-1,-2){10}}
\put(168,120){\line(1,-2){10}}
\put(232,120){\line(-1,-2){10}}
\put(248,120){\line(1,-2){10}}
\put(312,120){\line(-1,-2){10}}
\put(328,120){\line(1,-2){10}}
\put(25,195){\makebox(10,10)[t]{0}}
\put(25,170){\makebox(10,10)[t]{1}}
\put(25,145){\makebox(10,10)[t]{2}}
\put(25,120){\makebox(10,10)[t]{3}}
\put(25,90){\makebox(10,10)[t]{4}}
\put(360,195){\makebox(10,10)[t]{1$\Delta$}}
\put(360,170){\makebox(10,10)[t]{2$\Delta$}}
\put(360,145){\makebox(10,10)[t]{3$\Delta$}}
\put(360,120){\makebox(10,10)[t]{5$\Delta$}}
\put(360,90){\makebox(10,10)[t]{8$\Delta$}}
\end{picture}
\caption{The first 5 levels of the Stern-Brocot tree $\T(4/5,5/6)$.}
\label{fig:SB-tree}
\end{center}
\end{figure}
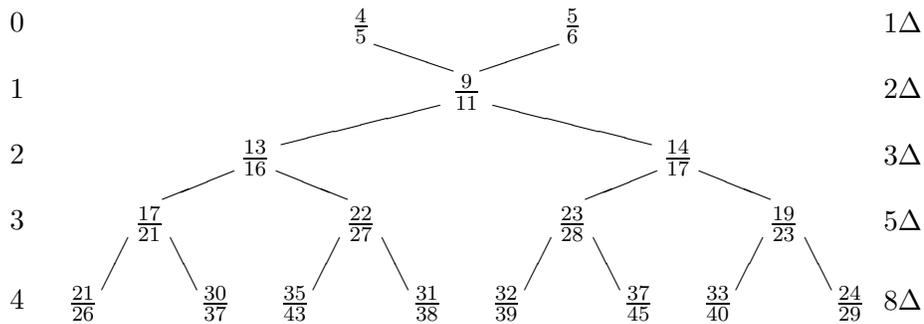

The following facts are either from \cite{Br1860,St1858}; or directly
from the definition of $\T(y_1/x_1, y_2/x_2)$; or can be shown by
easy induction:

\begin{enumerate}
\item All fractions in $\T(y_1/x_1, y_2/x_2)$ are distinct primitive vectors
and are strictly between $y_1/x_1$ and $y_2/x_2$. 

\item Each node in level $k$ is the result of the operator $\odot$
applied to a node in level $k-1$ and a node in level $\leq k-2$.

\item In each level $k$, there exists a node that is the result of the
operator $\odot$ applied to a node in level $k-1$ and a node in level $k-2$.

\item Let $\T_q$ be the subtree of $\T(y_1/x_1, y_2/x_2)$ from level 1
through level $q$. Let $\V_q$ be the set of the fractions contained 
in $\T_q$. Then $|\V_q| = 2^q-1$. 

\item For each node $y/x$ with the left child
$y'/x'$ and the right child $y''/x''$, we have $y'/x' < y/x < y''/x''$.
So the in-order traversal of $\T_q$
lists the fractions in $\V_q$ in increasing order.

\item Define the {\em size} of a node $y/x$ to be $\max\{x,y\}$.
The size of the nodes in level 0 (i.e., the two nodes $y_1/x_1$ and
$y_2/x_2$) is bounded by $1\cdot \Delta=F_1\cdot {\Delta}$ (because
both $y_1/x_1$ and $y_2/x_2$ are fractions in $\overline{P_{\Delta}}$).

\item The size of the node in level 1 (i.e., the node $y_3/x_3$)
is bounded by $2\cdot \Delta=F_2 \cdot \Delta$
(because $x_3 = x_1 + x_2 \leq 2\Delta$
and $y_3 = y_1 + y_2 \leq 2\Delta$).

\item For each $q\geq 2$, the size of level $q$ nodes is
bounded by $F_{q+1}\cdot \Delta$. (The last column in Fig \ref{fig:SB-tree}
shows the upper bounds of the size of the level $q$ fractions.)
\end{enumerate}

The lemma immediately follows from the above facts 1, 4 and 8.
\end{proof}

\begin{corollary}\label{cor:main}
Every $n$-vertex tree $T$ has a monotone drawing on a grid of size
at most $12n \times 12n$.
\end{corollary}
\begin{proof}
By Lemma \ref{lemma:valid}, $(3,3)$ is a valid-pair. Take this
valid-pair, the corollary follows from Theorem \ref{thm:main}.
\end{proof}

\section{Lower Bound}\label{sec:lower}

Let $T_0$ be a tree with root $r$ and 12 edge-disjoint paths
$P_1,P_2,...,P_{12}$, and each $P_i$ has $\frac{n}{12}$ vertices. 
In this section, we show that any monotone drawings of $T_0$ must use
a grid of size $\Omega(n) \times \Omega(n)$. Hence,
our result in Section \ref{sec:tree} is asymptotically optimal.

\begin{lemma}\label{lemma:lower_bound}
There exists a tree $T_0$ with $n$ vertices such that every monotone drawing
of $T_0$ must use an $\Omega(n) \times \Omega(n)$ grid.
\end{lemma}
\begin{proof}
Let $e_i$ ($1 \leq i \leq 12$) be the first edge in $P_i$
(see Fig \ref{Fig:lower_bound} (a)).
Let $\Gamma$ be any monotone drawing of $T_0$. Without loss of generality,
we assume the root $r$ is drawn at the origin $(0,0)$.

\begin{figure}[!htb]
\begin{center}
\includegraphics[width=0.75\textwidth, angle =0]{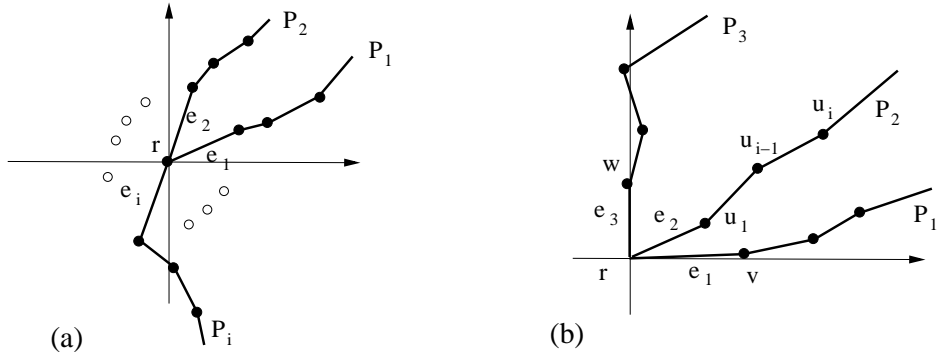}
\centering
\caption{(a) A tree $T_0$ with 12 edge-disjoint paths; (b) the proof
of Lemma \ref{lemma:lower_bound}.}
\label{Fig:lower_bound}
\end{center}
\vspace{-0.2in}
\end{figure}

By pigeonhole principle, at least three edges $e_i$ must be drawn
in the same quadrant. Without loss of generality, we assume the
edges $e_1,e_2$ and $e_3$ are drawn in the first quadrant in ccw order.
Let $e_1=(r,v)$, $e_3=(r,w)$, and $r=u_0,u_1, \ldots, u_k$ ($k=n/12$) be the
vertices of $P_2$. Thus $e_2 = (r,u_1)$, and $0^{\circ} \leq \Angle(r,v)
< \Angle(r,u_1) < \Angle(r,w) \leq 90^{\circ}$ (see Fig \ref{Fig:lower_bound} (b)). 

Consider the tree path $Q_1$ from $v$ to $u_i$ (for any index $1\leq i \leq k$).
Let $e^1_{\min}$ and $e^1_{\max}$ be the two extremal edges of $Q_1$ with
$\Angle(e^1_{\min}) < \Angle(e^1_{\max})$. Since $Q_1$  is monotone, we must have
$\Angle(e^1_{\max}) - \Angle(e^1_{\min}) < 180^{\circ}$ by Property
\ref{Prop:monotone-1}. 
Note that $\Angle(e^1_{\max}) \geq \Angle(v,r) \geq 180^{\circ}$.
This implies $\Angle(u_{i-1},u_i) \geq \Angle(e^1_{\min}) >
\Angle(e^1_{\max}) -180^{\circ}\geq 0^{\circ}$.

Now consider the tree path $Q_2$ from $u_i$ to $w$. 
Let $e^2_{\min}$ and $e^2_{\max}$ be the two extremal edges of $Q_2$ with
$\Angle(e^2_{\min}) < \Angle(e^2_{\max})$. Since $Q_2$ is monotone, we must have
$\Angle(e^2_{\max}) - \Angle(e^2_{\min}) < 180^{\circ}$ by Property
\ref{Prop:monotone-1}. 
Note that $\Angle(e^2_{\min}) \leq \Angle(r,w) \leq 90^{\circ}$.
This implies $\Angle(u_i,u_{i-1}) \leq \Angle(e^2_{\max}) <
\Angle(e^2_{\min}) +180^{\circ}\leq 270^{\circ}$.
Hence, $\Angle(u_{i-1},u_i) < 90^{\circ}$. 

Thus, for every edge $(u_{i-1},u_i) \in P_2$, $0^{\circ} < \Angle(u_{i-1},u_i)
< 90^{\circ}$. Let $(x(u_{i-1}),y(u_{i-1}))$ and $(x(u_i),y(u_i))$ be the two points 
in the drawing $\Gamma$ corresponding to $u_{i-1}$ and $u_i$, respectively. 
Because $0^{\circ} < \Angle(u_{i-1},u_i) < 90^{\circ}$,
we have $x(u_i) - x(u_{i-1}) \geq 1$ and $y(u_i) - y(u_{i-1}) \geq 1$.
So, in order to draw $P_2$, we need a grid of size at least
$\frac{n}{12} \times \frac{n}{12}$ in the first quadrant.
\end{proof}

\section{Conclusion}\label{sec:conc}
In this paper, we showed that any $n$-vertex tree has a monotone drawing on
a $12n \times 12n$ grid. The drawing
can be constructed in $O(n)$ time. We also described a tree
$T_0$ and showed that any monotone drawing of $T_0$ must
use a grid of size at least $\frac{n}{12} \times \frac{n}{12}$.
So the size of our monotone drawing of trees is asymptotically optimal.
 
It is moderately interesting to close the gap between the lower and the
upper bounds on the size of monotone drawing for trees. To reduce
the constant in the drawing size,
one possible approach is to improve Lemma \ref{lemma:valid},
whose proof is not tight.
In the Stern-Brocot tree $\T$, the sizes of the nodes near the leftmost and
the rightmost path of $\T$ are (much) smaller than the bound stated in
the proof of Lemma \ref{lemma:valid}. So it is possible to prove that
$(2^q -1 +t, F_{q+1})$ is a valid-pair for some integer $t$ ($t$ depends
on $q$). By using these better valid-pairs in Theorem \ref{thm:main},
it is possible to reduce the constant in the size of the drawing.


\begin{thebibliography}{4}

\bibitem{ACBFP12}
P. Angelini, E. Colasante, G. Di Battista,
F. Frati and M. Patrignani,
\emph{Monotone Drawings of Graphs},
J. of Graph Algorithms and Appl.
vol. 16, no. 1, pp. 5--35, 2012.

\bibitem{AFG10}
P. Angelini, F. Frati and L. Grilli,
\emph{An Algorithm to Construct Greedy Drawings of Triangulations},
J. of Graph Algorithms and Appl.
vol. 14, no. 1, pp. 19--51, 2010.

\bibitem{ADKM+13}
P. Angelini, W. Didimo, S. Kobourov,
T. Mchedlidze, V. Roselli, A. Symvonis, and S. Wismath,
\emph{Monotone Drawings of Graphs with Fixed Embedding},
Algorithmica, DOI 10.1007/s00453-013-9790-3, 2013.

\bibitem{ACM89}
E. M. Arkin, R. Connelly and J.S. Mitchell,
\emph{On Monotone Paths among Obstacles with Applications to Planning
Assemblies}, SoCG '89, pp. 334--343, 1989.

\bibitem{Br1860}
A. Brocot,
\emph{Calcul des Rouages par Approximation, Nouvelle Methode},
Revue Chronometrique, vol 6. pp. 186--194, 1860.



\bibitem{DT88}
G. Di Battista and R. Tamassia,
\emph{Algorithms for Plane Representations of Acyclic Digraphs},
Theor. Comput. Sci. vol. 61, pp. 175--198, 1988.



\bibitem{GT01}
A. Garg and R. Tammassia,
\emph{On the Computational Complexity of Upward and Rectilinear
Planarity Testing},
SIAM J. Comp. vol. 31 (2), pp. 601--625, 2001

\bibitem{HW89}
G. Hardy and E. M. Wright, 
\emph{An Introduction to the Theory of Numbers},
5th Edition, Oxford University Press, 1989.

\bibitem{HH15a}
X. He and D. He,
\emph{Compact Monotone Drawing of Trees},
in Proceedings of COCOON 2015, LNCS 9198,
pp. 457--468, 2015.

\bibitem{HH15b}
X. He and D. He,
\emph{Monotone Drawing of 3-Connected Plane Graphs},
in Proceedings of ESA 2015, LNCS 9294,
pp. 729--741, 2015.

\bibitem{HR14}
Md. Iqbal Hossain and Md. Saidur Rahman,
\emph{Monotone Grid Drawings of Planar Graphs},
in Proceedings of FAW 2014, LNCS 8497, pp. 105--116, 2014.

\bibitem{HEH2009}
W. Huang, P. Eades and S.H. Hong,
\emph{A Graph Reading Behavior: Geodesic-Path Tendency},
in Proceedings of IEEE Pacific Visualization Symposium,
pp. 137--144, 2009.

\bibitem{KSSW14}
P. Kindermann, A. Schulz, J. Spoerhase and A. Wolff,
\emph{On Monotone Drawings of Trees},
in Proceedings of GD 2014, LNCS 8871, pp. 488--500, 2014.

\bibitem{ML08}
A. Moitra and T. Leighton,
\emph{Some Results on Greedy Embeddings in Metric Spaces},
in Proceedings FOCS 2008, pp. 337--346, 2008.

\bibitem{NPR14}
M. N\"{o}llenburg, R. Prutkin, and I. Rutter,
\emph{On Self-Approaching and Increasing-Chord Drawings of 
3-Connected Planar Graphs},
in Proceedings GD 2014, LNCS 8871, pp. 476--487.

\bibitem{PR05}
C. H. Papadimitriou and D. Ratajczak,
\emph{On a Conjecture Related to Geometric Routing},
Theor. Comput. Sci., vol. 344 (1), pp. 3--14, 2005.

\bibitem{ST83}
D.D. Sleator and R.E. Tarjan,
\emph{A data structure for dynamic trees},
J. of Computer and System Sciences, vol. 24, pp. 362--391, 1983.

\bibitem{St1858}
M. A. Stern,
\emph{{\"U}ber eine Zahlentheoretische Funktion},
Journal f{\"u}r die reine und angewandte Mathematik vol. 55, pp. 193--220, 1958.


\end{thebibliography}
\end{document}